\documentclass[11pt]{article}
\normalsize
\usepackage{amsfonts}
\usepackage{amsmath}
\usepackage{amssymb}
\usepackage{url}

\newtheorem{theorem}{Theorem}

\newtheorem{definition}{Definition}

\newtheorem{remark}{Remark}
\newtheorem{example}{Example}

\newenvironment{proof}{\noindent {\em Proof.}}{\hspace*{\fill} $\Box $\newline}

\title{About Code Equivalence - a Geometric Approach}

\author{Iliya Bouyukliev and Stefka Bouyuklieva}
\date{}

\begin{document}

%
%
%


\maketitle

\begin{abstract}
The equivalence test is a main part in any classification problem. It helps to prove bounds for the main parameters of the considered combinatorial structures and to study their properties. In this paper, we present
 algorithms for equivalence of linear codes, based on their relation to multisets of points in a projective geometry.
\end{abstract}

%

\section{Introduction}\label{BB:Introduction}

 The problem for equivalence of linear codes is considered by many authors (see for example \cite{BB:bouyukliev2007code,BB:feulner2009automorphism,BB:NSendrier}). The most popular and widely used algorithms for code equivalence are based on
 the works of J. Leon \cite{BB:JLeon1982}. His programs are implemented in the software packages \textsc{Magma} \cite{BB:MAGMA} and \textsc{GAP} \cite{BB:GAP4}.
Leon's algorithm is very good for finding the automorphism
group, but doesn't give a canonical form. The canonical form is
very important for a fast comparison of a large number of objects
and it is the basis of algorithms for generating
combinatorial structures (see McKay \cite{BB:McKay2014}).

Our algorithms for equivalence are based on an algorithm for isomorphism of binary matrices. The set of all binary matrices with $n$ columns can be partitioned into equivalence classes under the action of the symmetric group $\mathcal{S}_n$. For each class of equivalence we choose one representative according to a certain rule, which we call a canonical representative. The canonical form of a matrix is the canonical representative of its equivalence class. The isomorphism test of matrices is then reduced to comparing their canonical forms.
In addition to the canonical forms, the algorithm used also calculates the order and generating elements of the automorphism groups of the considered matrices. For more details on this algorithm, called \textsc{IsBMiso}, see \cite{BB:bouyukliev2007code} and \cite{BB:LCequivalence}.

This article discusses the question in which cases the geometric approach in the linear code equivalence test is more effective, as well as why and how it can be applied. We compare the developed algorithms with the algorithm implemented in the program \textsc{LCequivalence} which is a module in the software package \textsc{QextNewEdition} \cite{BB:LCequivalence}.

The paper is organized in four sections. Section 2 consists of three subsections in which we give some important information about linear codes, Galois geometries and the relationship between the codes and multisets of points in a projective space.
 We also show an approach how to transform the
problem of code equivalence to the problem of isomorphism of
binary matrices. In Section 3, we describe the algorithms, named \textsc{CEIMPG} (Code Equivalence by Incidence Matrix of Projective Geometry) and \textsc{CESIMPG} (Code Equivalence by Shortened Incidence Matrix of Projective Geometry). In Section 4 we present some experimental results and compare the algorithms \textsc{CESIMPG} and \textsc{LCequivalence}.

\section{Preliminaries}\label{BB:Preliminaries}

Let $\mathbb{F}_q$ be a finite field with $q$ elements where $q=p^m$ for a prime $p$.
The \textit{support} $\mbox{\rm supp}(w)$ of a vector $w\in \mathbb{F}_q^n$ is the set of coordinate positions where the coordinates of $w$ are nonzero. The cardinality of the support is the \textit{Hamming weight} $\mbox{\rm wt}(w)$ of $w$, so it is equal to the number of its nonzero coordinates.
The \emph{Hamming distance\/} between two vectors of
$\mathbb{F}_q^n$ is defined as the number of coordinates in which they
differ.  A \emph{$q$-ary linear} $[n,k,d]_q$ {\it code} is a
$k$-dimensional linear subspace of $\mathbb{F}_q^n$ with minimum
distance $d$.  
Usually, a linear code is represented by its \textit{generator matrix}. The rows of a generator matrix form a basis of the code as a linear space. Here we use also a representation of the codes by their characteristic
vectors. 
For more details on the parameters and properties of linear codes we refer to \cite{BB:hufpless03}.

\subsection{Equivalence of linear codes}

\begin{definition}
We say that two linear $[n,k]_q$ codes $C_1$ and $C_2$ are
\textbf{equivalent}, if the codewords of $C_2$ can be obtained
from the codewords of $C_1$ via a finite sequence of
transformations of the following types: (1) permutation of
coordinate positions; (2) multiplication of the elements in a
given position by a non-zero element of $\mathbb{F}_q$; (3) application of
a field automorphism to the elements in all coordinate positions.
\end{definition}

This definition is well motivated as the transformations  (1)--(3)
preserve the Hamming distance and the linearity (for more details
see \cite[Chapter 7.3]{BB:KaskiOstergard}).
It is based on the action of the group ${\rm Mon}_n(\mathbb{F}_q)$ of all monomial $n\times n$ matrices for a prime
field and of the semidirect product ${\rm Mon}_n(\mathbb{F}_q)\rtimes {\rm Aut}(\mathbb{F}_q)$ for a composite field.

An \emph{automorphism} of a linear code $C$ is a pair $(M,\alpha)\in \mathrm{Mon}_n(\mathbb{F}_q)\rtimes \mathrm{Aut}(\mathbb{F}_q)$ such that $vM\alpha\in C$ for any codeword $v\in C$. The set of all automorphisms of the code $C$ forms the \emph{automorphism group} $\mathrm{Aut}(C)$.
For binary codes, $\mathrm{Aut}(C)$ consists only of permutation matrices and can be considered as a subgroup of the symmetric group $S_n$.

Many algorithms for codes use a set of codewords with given
properties - to be invariant with respect to the automorphism
group and to generate the code as a linear space. Usually, this
set consists of codewords with weights close to the minimum
weight. The problem of generating such a set is related to two other problems known as NP-complete -- the \textsc{Weight Distribution Problem}
\cite{BB:BMcET}, and the \textsc{Minimum
Distance Problem} \cite{BB:Vardy}. The complexity of
the \textsc{Code Equivalence Problem} is studied in \cite{BB:Petrank}.


We consider the \textsc{Code Equivalence Problem} for linear $q$-ary codes of length $n$. Leon's algorithm \cite{BB:JLeon1982} is based on the group action on a set of $(q-1)n$ points. The algorithm in the package \textsc{Q-Extension} and its successor \textsc{QextNewEdition} reduces the code equivalence problem to the problem for isomorphism of binary matrices with $2(q-1)n$ columns. A detailed description of this representation is given in \cite{BB:Maria2014}. In the main algorithm presented here, we use binary matrices with fewer, most often $n$ columns. For this reason, the presented algorithm is much more efficient in many cases.

\subsection{Galois geometries}

For the main definitions and theorems as well as more details on Galois geometries we refer to \cite[Section 14.4]{BB:FFHandbook} and \cite{BB:Storme2021}.

\begin{definition}
Let $V = V (n + 1, \mathbb{F})$, with $n\ge 1$, be an $(n + 1)$-dimensional vector space over
the field $\mathbb{F}$ with zero element 0. Define the equivalence relation $\sim$ on the set of nonzero vectors of $V$: for
$v,w\in V\setminus\{ 0\}$, $v\sim w$ if and only if $w = \alpha v$ for some $\alpha\in\mathbb{F}$, $\alpha\neq 0$.

(1) The set of equivalence classes is the \textit{$n$-dimensional projective space} over $\mathbb{F}$. It
is denoted by $\mathrm{PG}(n,\mathbb{F})$ or, when $\mathbb{F} = \mathbb{F}_q$, by $\mathrm{PG}(n,q)$.

(2) The elements of $\mathrm{PG}(n,\mathbb{F})$ are points; the equivalence class of the vector $X$ is
the point $[X]$. The vector $X$ is a coordinate vector for $[X]$ or $X$ is a vector
representing $[X]$. In this case, $\alpha X$ with $\alpha\in \mathbb{F}\setminus\{ 0\}$ also represents $[X]$, that is,
by definition, $[\alpha X] = [X]$.

(3) If $X = (x_0,\ldots,x_n)$ for some basis, then the $x_i$ are the coordinates of the point $[X]$.

(4) The points $[X_1],\ldots,[X_r]$ are linearly independent if a set of vectors
$X_1,\ldots,X_r$ representing them is linearly independent.
\end{definition}

\begin{definition}
Consider $V (n + 1,q)$ and its corresponding projective space $\mathrm{PG}(n,q)$.
For any $m = 0,1,\ldots,n$, an $m$-dimensional subspace, also called $m$-space, of
$\mathrm{PG}(n,q)$ is a set of points for which the union of all the corresponding coordinate
vectors, together with the zero vector, form an $(m + 1)$-dimensional vector subspace
of $V (n + 1,q)$. 
\end{definition}

 A $1$-dimensional subspace is called a (projective)
line, a $2$-dimensional subspace is called a (projective) plane, and a $3$-dimensional
subspace is called a (projective) solid. An $(n-1)$-dimensional subspace of $\mathrm{PG}(n,q)$
is called a hyperplane. An $(n-r)$-dimensional subspace of $\mathrm{PG}(n,q)$ is also called a
subspace of codimension $r$.

In $\mathrm{PG}(n,q)$, every hyperplane is a set of points $[X]$ whose coordinate
vectors $X = (x_0,\ldots,x_n)$ satisfy a linear equation
$$u_0x_0 + u_1x_1 + \cdots + u_nx_n = 0$$
with $u=(u_0,\ldots,u_n)\in\mathbb{F}_q^{n+1}\setminus\{(0,\ldots,0)\}$, and is denoted by $\pi(u)$.

\begin{definition}
 A collineation $\alpha$ of $\mathrm{PG}(n,q)$, $n\ge 2$, is a bijection which preserves
incidence. 
\end{definition}


\begin{theorem} {\rm (Fundamental Theorem of Galois Geometry)}
If $\alpha$ is a collineation of $\mathrm{PG}(n,q)$, $q=p^s$, $p$ prime, $s\ge 1$, then $\alpha$ is a semilinear bijective transformation
of $V (n + 1, q)$, i.e., there exists a nonsingular $(n + 1)\times (n + 1)$ matrix $A$ over $\mathbb{F}_q$ and an
automorphism $\rho : \mathbb{F}_q\to\mathbb{F}_q$, such that
\[
\alpha: \left[\begin{array}{c}
x_0\\
x_1\\
\vdots\\
x_n
\end{array}\right]\mapsto A\left[\begin{array}{c}
\rho(x_0)\\
\rho(x_1)\\
\vdots\\
\rho(x_n)
\end{array}\right]
\]
\end{theorem}

Let $\mathbb{Z}_{n+1}(q) =\{\delta : x \mapsto \alpha I_{n+1}x, \alpha\in \mathbb{F}_q\setminus\{0\}, x\in V(n+1,q)\}$, where
$I_{n+1}$ is the $(n + 1)\times (n + 1)$ identity matrix. The projective group of $\mathrm{PG}(n,q)$, $n\ge 1$, is the group
$\mathrm{PGL}_{n+1}(q) = \mathrm{GL}_{n+1}(q)/ \mathbb{Z}_{n+1}(q)$, and the collineation group of $\mathrm{PG}(n,q)$, $n\ge 1$, is the group
$\mathrm{P\Gamma L}_{n+1}(q) = \mathrm{PGL}_{n+1}(q) \rtimes \mathrm{Aut}(\mathbb{F}_q)$.

\begin{definition}
 Two sets $S$ and $S'$ of spaces contained in $\mathrm{PG}(n,q)$ are called projectively
equivalent to each other if and only if there is a collineation $\alpha\in \mathrm{P\Gamma L}_{n+1}(q)$ which
maps $S$ onto $S'$.
\end{definition}

\subsection{Linear codes and multisets of points}

The projective space $\mathrm{PG}(k-1,q)$ contains $\theta(k-1,q)=\frac{q^k-1}{q-1}$ points. For each point $[X]$, take $X$ to be the coordinate vector, whose first nonzero coordinate is 1 (we call such vectors \textit{normalized}), and then order the points lexicographically. We use this ordering to correspond a characteristic vector to each multiset $M$ of points in the projective geometry:
\begin{equation}\label{BB:chi_M}
\chi(M)=\left(\chi_1,\chi_2,\ldots,\chi_{\theta(k-1,q)}\right)\in \mathbb{Z}^{\theta(k-1,q)}
\end{equation}
where $\chi_u$ shows how many times the $u$-th point of $\mathrm{PG}(k-1,q)$ occurs in $M$, $u=1,\ldots,\theta(k-1,q)$.

There is a direct relation between the linear codes of dimension $k$ over $\mathbb{F}_q$ and the multisets of points in the projective geometry $\mathrm{PG}(k-1,q)$.
 Let $G$ be a generator
matrix of a $q$-ary linear $[n,k,d]$ code $C$, and let $g_1, g_2,\ldots, g_n\in\mathbb{F}_q^k=V(k,q)$
 be the columns of $G$. Suppose that none of these columns is the zero vector (then we say that the code $C$ is of full length). Each vector $g_i$ determines a point $[g_i]$ in the projective space $\mathrm{PG}(k-1,q)$. If the vectors $g_i$ are pair-wise independent, then $M_G = \{[g_1], [g_2],... , [g_n]\}$ is a
set of $n$ points in $\mathrm{PG}(k-1,q)$. When dependence occurs, we interpret $M_G$ as a multiset and
count each point with the appropriate multiplicity \cite{BB:DodunekovSimonis}.

On the other hand, if $M$ is a multiset of $n$ points in $\mathrm{PG}(k-1,q)$, the $k\times n$ matrix $G_M$, whose columns are the normalized coordinate vectors of the points from $M$, generates a linear $[n,k]$ code $C_M$. The minimum distance of $C_M$ is equal to $d$ if (a) each hyperplane of $\mathrm{PG}(k-1,q)$ meets $C_M$ in at most $n - d$ points and (b) there is a hyperplane meeting $C_M$ in exactly $n - d$ points. Some authors even give a definition for linear codes as multisets of points \cite{BB:Landjev}.
%
If the multiset $M$ is a set then we call $C_M$ a projective code. Two codes of full length
are equivalent if and only if the corresponding multisets of points are projectively
equivalent \cite{BB:DodunekovSimonis}.




The characteristic vector of the code $C$ with respect to its generator matrix $G$ is the characteristic vector of the multiset $M_G$, or
\begin{equation}\label{BB:chi}
\chi(C,G)=\left(\chi_1,\chi_2,\ldots,\chi_{\theta(k-1,q)}\right)\in \mathbb{Z}^{\theta(k-1,q)}
\end{equation}
where $\chi_u$ is the number of the columns of $G$ that are coordinate vectors of the $u$-th point of $\mathrm{PG}(k-1,q)$, $u=1,\ldots,\theta(k-1,q)$. 
When $C$ and $G$ are clear from the context, we will briefly write $\chi$.

A code $C$ can have different characteristic vectors depending on the chosen generator matrices. If we permute the columns of the matrix $G$ we will obtain a permutation equivalent code to $C$ having the same characteristic vector. Moreover, from a characteristic vector one can restore the columns of the generator matrix $G$ but eventually at different order and/or multiplied by  nonzero elements of the field.

\section{The algorithms}
\label{BB:algorithm}

\subsection{Algorithm for Code Equivalence using the Incidence Matrix\\ of Projective Geometry}


Denote by $G_{q,k}$ the $k\times \theta(k-1,q)$ matrix whose columns are the normalized coordinate vectors of the points in $\mathrm{PG}(k-1,q)$ ordered lexicographically. The rows of $G_{q,k}$ are linearly independent and so it generates a $q$-ary linear code of dimension $\theta(k-1,q)$ and dimension $k$. This code is called the \textit{simplex code} and denoted by $\mathcal{S}_{q,k}$. Its characteristic vector is $(1,1,\ldots,1)$, and all nonzero codewords of $\mathcal{S}_{q,k}$ have weight $\theta(k-1,q)/2$.

Further, we consider the matrix $A_k=G_{q,k}^{\rm T}\cdot G_{q,k}$. The rows of this matrix form a maximal set of nonproportional codewords in the considered simplex code. For the elements of $A_{k}$ we have $a_{ij}=u_i\cdot u_j=\sum_{m=1}^k u_{mi}u_{mj}$, where $u_i\cdot u_j$ is the Euclidean inner product of the vectors $u_i,u_j\in \mathbb{F}_q^k$ over the field $\mathbb{F}_q$. Obviously, the $i$-th row of the matrix $A_{k}$ can be identified with the hyperplane $\pi(u_i)$ with the following equation
$$u_{1i}x_1+u_{2i}x_2+\cdots+u_{ki}x_k=0.$$
We denote by ${\cal N}(A_k)$ the matrix obtained from $A_k$ by replacing all nonzero elements by $1$ and call it a normalized matrix. Obviously, ${\cal N}(A_k)_{ij}=0$ if and only if the vectors $u_i$ and $u_j$ are mutually orthogonal. This can be interpreted in the following way: ${\cal N}(A_k)_{ij}=0$ if and only if the point $[u_j]$ is incident with the hyperplane induced by $u_i$. If we juxtapose 0's and 1's in ${\cal N}(A_k)$, we obtain the incidence matrix of the points and hyperplanes in the projective space $\mathrm{PG}(k-1,q)$. Both matrices have the same automorphism group which we denote by $\mathrm{Aut}_k=\mathrm{Aut}({\cal N}(A_k))$. It consists of the permutations of the columns that preserve the set of rows of the matrix, so it is a subgroup of the symmetric group $S_{\theta(k-1,q)}$. 
Since the automorphism group of a finite incidence structure acts as permutation group on the points, the automorphism group of the matrix ${\cal N}(A_k)$ is isomorphic to $\mathrm{P\Gamma L}_{k}(q)$ and instead of acting on the points in $\mathrm{PG}(k-1,q)$ or on the columns of the matrix ${\cal N}(A_k)$, we can take the action on the characteristic vectors.

Let $\mathfrak{C}$ be the set of the projective $[n,k]_q$ codes. Consider the action of the group $\mathrm{Aut}_k$ on the set $\mathfrak{C}$ as the codes are represented by their characteristic vectors. Then we have the following theorem.


\begin{theorem}
Two projective linear $[n,k]_q$ codes are equivalent if and only if their characteristic vectors belong to one orbit under the action of $\mathrm{Aut}_k$ on the set $\mathfrak{C}$.
\end{theorem}

\begin{proof}
According to \cite[Proposition 1]{BB:DodunekovSimonis}, two codes of full length
are equivalent if and only if the corresponding multisets of points are projectively
equivalent. On the other hand, two multisets $M_1$ and $M_2$ are equivalent if and only if there is a permutation $\pi$ that maps the points of $M_1$ to the points of $M_2$, which can be applied further to the characteristic vectors. It turns out that the  multisets $M_1$ and $M_2$ are equivalent if and only if there is a permutation $\pi$ such that $\chi(M_1)\pi=M_2$ (the characteristic vectors belongs to one orbit).
\end{proof}

Take $C_i=C_{M_i}$ and $G_i$ to be the matrix whose columns are the normalized coordinate vectors of the points in $M_i$, $i=1,2$.
Since the groups $\mathrm{P\Gamma L}_{k}(q)$ and $\mathrm{Aut}_k$ are isomorphic, to the permutation $\pi$ we can correspond a nonsingular matrix $T_\pi\in \mathrm{PGL}(k,q)$ such that $T_\pi G_1= G_2P$ where $P$ is a permutation matrix that permutes the columns of the $k\times n$ matrix $G_2$.

We present the algorithm for projective linear codes which means that their corresponding multisets are actually sets of points, and the coordinates of their characteristic vectors are only 0's and 1's.

We apply our program for isomorphism of binary matrices to obtain the inequivalent codes in the following way:
\begin{enumerate}
\item to any code $C$ with a characteristic vector $\chi$ with respect to its generator matrix we correspond the $(\theta(k-1,q)+1)\times\theta(k-1,q)$ matrix $G_C=\left(\begin{array}{c} {\cal N}(A_k)\\ \chi \end{array}\right)$;
\item we run the isomorphism test \textsc{IsBMiso} for these  matrices.
\end{enumerate}

\begin{remark} If the codes are not projective, we can add coloring of the columns and then apply the same algorithm.
\end{remark}

\subsection{Algorithm for Code Equivalence using a Shortened Incidence Matrix of Projective Geometry}

We will describe the algorithm in the case when $q$ is a prime. We use smaller binary matrices and prove that if these matrices are not isomorphic then the codes are not equivalent. If these matrices are isomorphic, we use a special approach to see if the corresponding codes are equivalent.

Instead of the whole matrix $\mathcal{N}(A_k)$ we take only those columns that correspond to nonzero coordinates of the considered characteristic vector $\chi(C,G)$ of the code $C$. The same matrix can be obtained by normalizing the matrix $A_G=A_k^TG$. The rows of $A_G$ are not proportional to each other and they are nonzero codewords in $C$. Furthermore, any nonzero codeword in $C$ is proportional to a row-vector of this matrix. Without loss of generality, we can take the matrix $G$ in systematic form which means that all column-vectors of weight 1 are columns in $G$. Moreover, we can take the columns to be normalized vectors.

\begin{theorem}\label{BBlem:ncong}
If the matrices $\mathcal{N}(A_{G_1})$ and $\mathcal{N}(A_{G_2})$ are not isomorphic then the codes $C_1$ and $C_2$ with generator matrices $G_1$ and $G_2$ are not equivalent.
\end{theorem}

\begin{proof} Let $\mathcal{N}(A_{G_1})\ncong \mathcal{N}(A_{G_2})$. Then the matrices $\mathcal{N}(G_{C_1})$ and $\mathcal{N}(G_{C_2})$ are not isomorphic, so the codes $C_1$ and $C_2$ are not equivalent.
\end{proof}

Theorem \ref{BBlem:ncong} shows that the case $\mathcal{N}(A_{G_1})\ncong \mathcal{N}(A_{G_2})$ is clear. Let us now consider the opposite case, when $\mathcal{N}(A_{G_1})\cong \mathcal{N}(A_{G_2})$. Denote the automorphism group of $\mathcal{N}(A_{G_1})$ by $H_1$. Recall that the algorithm \textsc{IsBMiso} tests the binary matrices for isomorphism and  computes generating elements of their automorphism groups.
There are two possibilities for $H_1$ - (1) to be trivial or (2) to contain at least two elements.

\bigskip
(1) Let $H_1=\{ \epsilon\}$.

Since the group is trivial, there is a unique permutation $\sigma\in S_n$ that maps the rows of $\mathcal{N}(A_{G_1})$ into the rows of $\mathcal{N}(A_{G_2})$. The codes are equivalent if there is a monomial matrix $M_{\sigma}=P_{\sigma}D$ such that $G_1M_{\sigma}$ generates the second code, where $P_{\sigma}$ is the permutation matrix corresponding to $\sigma$, and $D=\mathrm{diag}(\lambda_1,\ldots,\lambda_n)$ is a nonsingular diagonal matrix. Since $G_2$ and $G_1M_{\sigma}$ generate the same code, there is an invertible matrix $Q\in\mathrm{GL}(k,q)$ such that $QG_2=G_1M_{\sigma}$. Without loss of generality we can consider $G_2$ in the form $(I_k \vert E)$. Hence
\begin{equation}\label{BB:equ}
QG_2=(Q\vert QE)=G_1M_{\sigma}=
G_1P_{\sigma}\mathrm{diag}(\lambda_1,\ldots,\lambda_n).
\end{equation}

The columns of $G_1P_{\sigma}$ are the permuted columns of $G_1=(g_1,g_2,\dots,g_n)$. It turns out that $G_1P_{\sigma}D=(\lambda_1 g_{i_1},\ldots,\lambda_n g_{i_n})$ where $\sigma^{-1}(s)=i_s$, $s=1,\ldots,n$. Hence $Q=(\lambda_1 g_{i_1},\ldots,\lambda_k g_{i_k})$.
The next step is to solve the system of linear equations
$$QE=(\lambda_1 g_{i_1},\ldots,\lambda_k g_{i_k})E=(\lambda_{k+1} g_{i_{k+1}},\ldots,\lambda_n g_{i_n})$$
with variables $\lambda_1,\ldots,\lambda_n$. This is a homogeneous system so it is consistent, but we are looking for a solution in which all entries are nonzero. Such a solution gives an invertible matrix $Q\in\mathrm{GL}(k,q)$ which maps the first set of points into the second one and then these multisets and their codes are equivalent. It also means that there is an automorphism of ${\cal N}(A_k)$ that maps the characteristic vector of the code $C_1$ to the characteristic vector of the second code.

If no solution has the needed property, the codes are inequivalent.

\bigskip
(2) Let $H_1$ is not trivial.

Let $H_1=\left\langle \tau_1,\ldots,\tau_m\right\rangle$. We can use two approaches in the algorithm. In the first one, for each permutation $\tau_i$, we are looking for a nonsingular matrix $Q_i$ such that
\begin{equation}\label{BB:eq_i}
Q_iG_1=G_1P_{\tau_i}D_i =G_1P_{\tau_i}\mathrm{diag}(\alpha_{i1},\ldots,\alpha_{in})\end{equation}
for nonzero elements $\alpha_1,\ldots,\alpha_n\in \mathbb{F}_q$. If we have needed solutions of the considered systems of linear equations and have computed the invertible matrices $Q_i$, we go to the last step. The existence of nonsingular matrices $Q_i$ for all $i$ shows that the order of $\mathrm{Aut}(C_1)$ is equal to $(q-1)|H_1|$. The matrices $Q_i$ generate the automorphism group of the corresponding set of points. For the last step, we need a permutation $\sigma\in S_n$ that maps the rows of $\mathcal{N}(A_{G_1})$ into the rows of $\mathcal{N}(A_{G_2})$. As in the case (1), we are looking for an invertible matrix $Q$ such that the system $$QG_2=(Q\vert QE)=G_1P_{\sigma}\mathrm{diag}(\lambda_1,\ldots,\lambda_n)$$
has a solution in which all entries are nonzero.

If for some $\tau_i$ the system (\ref{BB:eq_i}) does not have a solution with nonzero entries, then we go to the first algorithm.

\begin{example}
Consider the ternary codes $C_1$ and $C_2$ with generator matrices
$$G_1=\left[\begin{array}{cccccc}
1&0&0&1&2&0\\
0&1&0&1&1&1\\
0&0&1&1&1&0\\
\end{array}\right], \ \ \ G_2=\left[\begin{array}{cccccc}
1&0&0&1&1&0\\
0&1&0&1&2&0\\
0&0&1&1&0&2\\
\end{array}\right],
$$
respectively. The permutation $\sigma=(2 \ 3 \ 4)$ maps  the rows of $\mathcal{N}(A_{G_1})$ into the rows of $\mathcal{N}(A_{G_2})$. The automorphism groups of the two matrices are not trivial but we go directly to the last step of the algorithm, so we are looking for an invertible matrix $Q\in \mathrm{GL}(3,3)$ such that $$QG_1=G_2P_{\sigma}\mathrm{diag}(\lambda_1,\dots,\lambda_6).$$
This gives the following system of linear equations:
$$\begin{array}{|rrrl}
  \lambda_1&+\lambda_2 && =0 \\
 &\lambda_2&+\lambda_3 & =0\\
 &\lambda_2& & =\lambda_4\\
\end{array} \ \ \ \ \
\begin{array}{|rrrl}
  \lambda_1&+2\lambda_2 && =2\lambda_5 \\
 &2\lambda_2&& =\lambda_5\\
 &2\lambda_2& & =\lambda_5\\
\end{array} \ \ \ \ \
\begin{array}{|rrrl}
  & &0& =0 \\
 &&2\lambda_3 & =\lambda_6\\
 &&0 & =0\\
\end{array}.$$
The solution is $p(1,2,1,2,1,2)$, $p\in \mathbb{F}_3$. Thus we obtain
$$Q=\left[\begin{array}{ccc} 1&2&0\\ 0&2&1\\ 0&2&0\end{array}\right].$$
It turns out that the two codes are equivalent.
\end{example}

The presented example shows that we can prove that two codes are equivalent using only the last step in the algorithm. The problem arises when there is no invertible matrix to satisfy (\ref{BB:equ}). This fact does not prove the inequivalence of the considered codes and therefore we have to follow the other steps of the algorithm. It is possible to obtain an invertible matrix that sends the codewords of $C_1$ to codewords of $C_2$ from some of the permutations $\sigma\tau_i$, $1\le i\le m$. Since we use canonical forms in the program for isomorphism of binary matrices, even if neither of the permutations $\sigma,\sigma\tau_1,\ldots,\sigma\tau_m$ produces invertible matrix, the equivalence of the two codes is still possible. Therefore, in such a situation we use the first algorithm, namely \textsc{CEIMPG}.

The second approach has the disadvantage that we do not count the automorphism groups of the codes. If we compare only two codes, this is not important, but if we have a set with more than a thousand codes, the first approach is more useful.

\begin{remark}
If the field is composite (when $q=p^m$, $p$ - prime, $m>1$) then the matrix equation (\ref{BB:equ}) changes to
\begin{equation}\label{BB:equ_q}
QG_2=(Q\vert QE)=
G_1P_{\sigma}\mathrm{diag}(\lambda_1,\ldots,\lambda_n)\rho,
\end{equation}
where $\rho$ is an automorphism of the field $\mathbb{F}_q$. In this case there are $n$ unknown variables $\lambda_1,\ldots,\lambda_n$ and an unknown automorphism $\rho$. Recall that the automorphism group of the field $\mathbb{F}_q$ is a cyclic group of order $m$.
\end{remark}

\section{Experimental results}
\label{BB:experimental}


The number of the needed basic operations in the described algorithms depends on the size of the input data and the structure of the considered codes. There is a relationship between the structure of the binary matrices that are used in the algorithms \textsc{CESIMPG} and \textsc{LCequivalence}. If the matrices correspond to regular combinatorial structures,
such as orthogonal arrays, t-designs or Hadamard matrices, the algorithms need more operations to compute the automorphism groups, to obtain the canonical forms and to distinguish the inequivalent codes. In fact, the difference in computational time between the algorithms comes from the difference in the size of the input data. For example, if we consider $[24,4]$ codes over $\mathbb{F}_3$, the algorithm \textsc{CESIMPG} uses $40\times 24$ binary matrices, but the matrices in \textsc{LCequivalence} have size $s\times 84$ where $s$ is the number of codewords in the considered generating set of the code, so we can expect that the first algorithm will be faster. Comparing the sizes, we conclude that presented here algorithm is faster for small dimensions. We present some experimental results in Table \ref{BB:table-1}. We first generate random codes with given length and dimension (their number is shown in column 3), then check them for equivalence. The number of inequivalent codes is given in column 4. In the last two columns we present the computational time of algorithms \textsc{CESIMPG} and \textsc{LCequivalence}, respectively.

All examples are executed on  (\textsc{Intel Core  i7-6700HQ 2.60 GHz processor}) in Active solution configuration --- Release, and Active solution platform --- \textsc{X64}. As a development environment for both algorithms we use \textsc{MS Visual Studio 2019}.

In addition, we have to mention that \textsc{CESIMPG} can be further improved in several directions. For example, for larger fields, many of the rows in the matrix $A_G$ have maximum supports, i.e. their Hamming weights are equal to the length of the code. After normalization, they go into the all-ones vector which does not give any information about the automorphism group and the orbits, and therefore we can remove these rows. The work on this algorithm is still ongoing and we expect to have better results in the computational time.

\begin{table}
\caption{Experimental results}
{\begin{tabular}{c|c|c||r|r|r|r}
\hline\noalign{\smallskip}
&& &$\sharp$ generated & $\sharp$ inequivalent  &    &   \\
$q$&$k$&$n$  & codes & codes   &  \textsc{CESIMPG}   & \textsc{LCequivalence}   \\
\noalign{\smallskip}\hline\noalign{\smallskip}
3 & 3 &10 &10 000 &   347  & 0.59s   & 2.68s \\
  & 3 &24 &10 000 & 8 306  & 1.16s   & 71.58s \\
  & 4 &10 &10 000 & 1 275  & 1.05s   & 2.17s \\
  & 4 &24 &10 000 & 10 000 & 1.87s   & 14.47s \\
  & 5 &10 &10 000 & 1 946 & 1.84s   & 2.17s \\
  & 5 &24 &10 000 & 10 000 &  3.24s &  8.78s\\
\hline
7 & 3 &10 &10 000 & 8 288  &  1.81s & 15.43s \\
  & 3 &24 &10 000 & 10 000 & 1.75s  & 136s \\
  & 4 &10 &1 000 &  999   &  0.56s  & 2.38s\\
  & 4 &24 &1 000 & 1 000  & 0.65s   & 7.55s \\
  & 5 &10 &1 000 & 1 000  & 4.79s   & 1.69s \\
  & 5 &24 &1 000 & 1 000  & 3.79s   & 5.88s \\
  \hline
11 & 3 &10 &10 000 & 9 986  &  5.69s   & 37.57s \\
  & 3 &24 &10 000 & 10 000  &  2.56s   &  335.69s\\
  & 4 &10 &1 000 & 1 000  &   2.38s  &  5.21s\\
  & 4 &24 &1 000 & 1 000  &  1.74s   &  21.80s\\
  & 5 &10 &1 000 & 1 000  &   27.96s  & 3.95s \\
  & 5 &24 &1 000 & 1 000  & 30.84s    &16.00s \\
\noalign{\smallskip}\hline
\end{tabular}}
\label{BB:table-1}
\end{table}

\section*{Acknowledgements}

The research of Stefka Bouyuklieva was supported by a Bulgarian NSF contract KP-06-N32/2-2019.
The research of Iliya Bouyukliev was supported, in part, by a Bulgarian NSF contract KP-06-Russia/33/17.12.2020.


\end{document}